%% file: main.tex
\documentclass[letterpaper, 10 pt, conference]{ieeeconf}  

\IEEEoverridecommandlockouts                              

\usepackage{cite}
\usepackage{amsmath,amssymb,amsfonts}
\usepackage{algorithmic}
\usepackage{graphicx}
\usepackage{algorithm,algorithmic}
\usepackage{hyperref}
\usepackage{textcomp}
\usepackage[thinc]{esdiff}
\usepackage{xcolor}
\usepackage{multirow}

\newtheorem{theorem}{Theorem}

\newtheorem{definition}[theorem]{Definition}
\newtheorem{ass}[theorem]{Assumption}
\newtheorem{remark}[theorem]{Remark}

\renewcommand{\Re}{\mathbb{R}}

\newcommand{\U}{\mathcal{U}}

\newcommand{\umax}{u_{\max}}

\DeclareMathOperator*{\minimize}{minimize} 		
\DeclareMathOperator*\stt{subject\ to}%

\title{\LARGE \bf
A Scenario-based Model Predictive Control Scheme for Pandemic Response through Non-pharmaceutical Interventions
}

\author{Domagoj Herceg$^{1}$, Marco Dell'Oro$^{1}$, Riccardo Bertollo$^{1}$, Fuminari Miura$^{2}$ \\
Paul de Klaver$^{3}$, Valentina Breschi$^{4}$, Dinesh Krishnamoorthy$^{1}$, Mauro Salazar$^{1}$
\thanks{*This work was supported by ZonMw grant no. 10430362220008.}
\thanks{$^{1}$The authors are with the Department of Mechanical Engineering,
        Eindhoven university of Technology, Eindhoven, The Netherlands. \newline
        {email: \tt\small \{d.herceg,m.delloro,r.bertollo,\newline d.khrishnamoorthy, m.r.u.salazar\}@tue.nl}}%
\thanks{$^{2}$Fuminari Miura is with the Centre for Infectious Disease Control, National Institute for Public Health and the Environment (RIVM),
Bilthoven, the Netherlands, {email: \tt\small fuminari.miura@rivm.nl}}%
\thanks{$^{3}$Paul de Klaver is with the Maxima Medisch Centrum, Eindhoven, The
        Netherlands {\tt\small p.deklaver@mmc.nl}}%
\thanks{$^{4}$Valentina Breschi is with the Department of Electrical Engineering,
        Eindhoven university of Technology, Eindhoven, The Netherlands. \newline
        {email: \tt\small v.breschi@tue.nl}}%
}

\begin{document}

\maketitle
\thispagestyle{plain}
\pagestyle{plain} 

\begin{abstract}
    This paper presents a scenario-based model predictive control (MPC) scheme designed to control an evolving pandemic via non-pharmaceutical intervention (NPIs). The proposed approach combines predictions of possible pandemic evolution to decide on a level of severity of NPIs to be implemented over multiple weeks to maintain hospital pressure below a prescribed threshold, while minimizing their impact on society. 
    Specifically, we first introduce a compartmental model which divides the population into Susceptible, Infected, Detected, Threatened, Healed, and Expired (SIDTHE) subpopulations and describe its positive invariant set. This model is expressive enough to explicitly capture the fraction of hospitalized individuals while preserving parameter identifiability  w.r.t.  publicly available datasets.
    Second, we devise a scenario-based MPC scheme with recourse actions that captures potential uncertainty of the model parameters. e.g., due to population behavior or seasonality.
    Our results show that the scenario-based nature of the proposed controller manages to adequately respond to all scenarios, keeping the hospital pressure at bay also in very challenging situations when conventional MPC methods fail.
\end{abstract}


\input{Sections/intro.tex}
\input{Sections/compartmental.tex}

\input{Sections/sidthe.tex}
\input{Sections/nmpc.tex}
\input{Sections/numerical.tex}
\input{Sections/conclusions.tex}




\section*{ACKNOWLEDGMENTS}
We thank Dr.\ I.\ New for proofreading this paper.


\bibliographystyle{IEEEtran}
\bibliography{mppc} 

\end{document}

%% file: Sections/intro.tex
\section{INTRODUCTION}\label{sec:introduction}

The latest COVID-19 pandemic has unveiled the fragility of the healthcare system worldwide, as well as the limitation of existing strategies at the disposal of policymakers to face a pandemic outbreak, namely vaccinations and non-pharmaceutical interventions (NPIs). 
Fueled by the need to respond to pandemic outbreaks in a more timely manner, several researchers have thus focused on establishing approaches to respond to potential future pandemics either through vaccine allocation (see, e.g., \cite{Miura_optimal,tonkens_vaccine,calafiore_vaccine}) or scheduling NPIs~\cite{Britton,optiml_bliman}.  

Although the latter kinds of countermeasures have proven particularly effective in practically containing virus transmissions~\cite{KUCHARSKI20201151,NPI_effectivnes}, they impose a severe burden on society~\cite{economicImpact_1,economicImpact_2} hence, their use should be as minimal as possible. Along this line, the study in~\cite{optiml_bliman} highlights the optimality of implementing a single NPI with the highest possible severity, simultaneously highlighting the pivotal role played by the timing with which such an NPI is enacted. In particular, \cite{Britton} shows that the optimal timing of interventions occurs when the infected subpopulation hits approximately $20\%$ of the total. Such an optimal timing is thus non-implementable in a real-world scenario, as the prevalence (i.e., the total number of infected individuals at a given time point) at that epidemic stage would overwhelm the healthcare system of any country. Another line of research focuses on limiting the peak prevalence at any time of the pandemic. The central objective of NPIs is akin to what is colloquially known as \emph{flattening the curve}, i.e., keeping the prevalence under a certain threshold for the healthcare system to assist individuals needing medical support. With this strategy, policymakers decide on the moment to implement more stringent control measures, and once the epidemic starts to decline as a consequence of the set of NPIs, they often lift those measures to balance societal activities during the COVID-19 pandemic. Markedly, this strategy is much more relevant in the practice of outbreak management as the critical threshold can be defined by the total capacity of the healthcare system of the country of interest. 
Although there have been approaches that try to combine both strategies mentioned above~\cite{Ferramosca}, these still fall short of being implementable in a real-life setting. First, existing studies often investigate this problem theoretically, relying on naive assumptions such as overly simple models, known parameters, and linear cost criteria. Once computed, NPI measures are assumed to be implemented for the duration of the pandemic
without accounting for possible re-adaptation of interventions.

A more relevant approach, considering the need for continuous adaptation in a predictive fashion, would be the so-called Model Predictive Control (MPC)~\cite{MPC_book} strategy. In this feedback framework, a control action is recalculated
periodically to reflect the latest available information and to account for potential unforeseen disturbances.
Some studies have already adapted MPC and \emph{MPC-like} approaches~\cite{allgower_MPC_covid,MPC_noteminalcost,MPC_teminalcost,N-Step-Ahead,NMPC_integer} to pandemic
management. 
Along the same lines of these works, we argue that MPC can be the standard reference approach in pandemic control, as it mimics what the NPI strategies aim to achieve: 
assessing the current state of the pandemic, suggesting interventions based on the predicted evolution, and reassessing and refining their recommendations based on the latest available data and insights. 
Other key points in favor of MPC for pandemic control lay in its feedback nature, introducing robustness in the closed loop, and its capability to incorporate constraints on the states and control actions. However, even if featuring these aspects, most existing MPC approaches for pandemic control still fail to account 
for the uncertainty of the temporal evolution of an epidemic and the possibility of different intervention strategies due to the differential epidemic progression. 


The common ground between virtually all control approaches for pandemic management, and particularly for MPC-based strategies, is modeling the pandemic evolution by \emph{compartmental models}. For an overview of compartmental models, we point the reader to~\cite[Chapter 2]{compartmental_models}.
In recent years, different authors have proposed extensions and improvements of the original SIR~\cite{SIR} model, while trying to maintain its simplistic nature \cite{Gibson2020, Grimm2021, Anastassopoulou2020 }.
The main limitation of the aforementioned models is the absence of a compartment able to capture the number of admissions to hospitals or intensive care units (ICU), which is arguably the most critical factor in a pandemic outbreak.
To address this shortcoming, the SIDARTHE model~\cite{SIDARTHE} has been developed to accurately describe the spread of COVID-19 in Italy, inspiring
further development of other complex compartmental models, such as \cite{Papageorgiou2023, Carli2020}. The main issue with these very detailed models is that parameter identifiability is rarely guaranteed (or even considered), ultimately not allowing the estimation of the parameters that characterize these models from the data and thus requiring the use of \textquotedblleft standard\textquotedblright \ values (potentially not validated). This raises the question of how reliable these models can be in predicting the future behavior of pandemics and informing policy decisions. 

To address these limitations of the current state-of-the-art for pandemic control, we i) propose the SIDTHE compartmental model, which retains the compartments that are critical in pandemic management (i.e., those related to hospital/ICU patients) while being practically identifiable with open-source pandemic data. On the modeling side, SIDTHE strikes a balance between expressiveness and identifiability compared to other models found in the literature. Then, toward pandemic control, we ii)  analyze a positive invariant set for the SIDTHE model, which is a key property to ensure safety in pandemic management. Note that, looking at pandemic models, invariant sets have been described for SIR and SERI models in~\cite{epidemic_robust_admissible_sets}.

Lastly, we iii) introduce \emph{scenario-based}\footnote{A \emph{scenario} refers to a particular realization of the uncertainty.} MPC approach for pandemic management. This choice allows us to go a step further with respect to existing strategies for control-based pandemic management. Indeed, 
by accounting for many scenarios, we can describe a whole gamut of possible realizations of the uncertainties affecting pandemic control problems and account for them explicitly. Moreover, we can compute different control strategies for each
particular scenario (uncertainty realization) subject to causality constraints~\cite{min-max-MPC}.
This allows us to account for possible imprecision in the prediction model describing the pandemic evolution, as well as being robust against uncertainty in the effects of the proposed NPIs. 



The remainder of the paper is structured as follows: Section II briefly reviews the basic assumptions and properties of compartmental models. In Section III, we present our SIDTHE model and describe its positive invariant set. 
We then introduce the scenario-based stochastic NMPC approach (see Section IV) and present a numerical case study in Section V. Conclusions are drawn in Section VI, with an outlook on possible directions for future work.

\paragraph*{Notation}
The set of real vectors of dimension $n$ 
is denoted with $\Re^n$. A set of real, non-negative (positive) numbers is denoted with $\Re^n_+$ $(\Re^n_{++})$. 
As a convention, we shall drop the superscript when the dimension $n$ 
is equal to one and use the (double) minus subscript when dealing with (negative) non-positive numbers.

%% file: Sections/compartmental.tex
\section{Compartmental Models in Epidemiology}
For the purpose of introduction, we consider a generic $n$ state compartmental 
model of the pandemic given by a set of $n$ ordinary differential equations (ODEs) 
succinctly written as
\begin{align}
\frac{\mathrm{d} x(t)}{\mathrm{d} t} = f(x(t),\theta(t)),
\label{eq:dynamics_auto}
\end{align}
where $x(t) \in \Re^n_{+}$ is the non-negative state vector, 
$\theta \in \Re^{m}_{++}$ is a vector of strictly positive parameters and the transition function $f : \Re^n \times \Re^m \rightarrow \Re^n$ satisfies the following. 

\begin{ass}
    The function $f : \Re^n \times \Re^m \rightarrow \Re^n$ in~\eqref{eq:dynamics_auto} is continuously differentiable, 
    both in states and in parameters.
\end{ass}
Each of the $n$ ODEs corresponds to a subpopulation of interest. In general, these can be divided into a source, usually referred to as susceptible, a set of transient nodes, related to infected individuals, and the accumulators (sink or drain nodes) such as healed or people who have passed away during the pandemic. Accumulator nodes have no outflows, while a single source node has no inflows. The most simple model in this nomenclature is the SIR model. The most considered compartment is the number of infected (or infectious) people which we will denote with $I$,  where $I(t)$ for $t \ge 0$ describes the evolution of infected individuals as a part of the population. Other names will be identified as it is opportune in the remainder of the paper.

Compartmental models further satisfy mass conservation, that is,
$\sum_i \dot{x}_i(t) = 0$. As a consequence,  
$\sum_i^n x_i(t) = 1 , \quad \forall t \in \Re_{+}$,
with states scaled to always sum to one, which is the standard convention. This property implies that the population remains constant, that is, there is
no outflow of people (emigration) and no inflow (no immigration and new births are negligible). Meanwhile, it implies that the state trajectory is constrained to lie in a (probability) simplex
\begin{align}
\label{eq:prob_simplex}
    \Delta_{n} = \{ x \in \Re^n \mid 0 \le x \le 1, 1^\top  x = 1  \}.
\end{align}
 
In pandemic control, we assume that a policymaker can influence the evolution of the pandemic by exerting a certain control action (policy) $u \in \U = [0,\umax]$.
The dynamics in \eqref{eq:dynamics_auto} is thus further extended as:
\begin{align}
    \frac{\mathrm{d} x(t)}{\mathrm{d} t}  = f(x(t),u(t), \theta(t)).
    \label{eq:dynamics_control}
\end{align}
Moreover, the control variable is usually related to a decrease in the parameter tied to the disease spread. This is also the case in our paper, as we will explain in Section~\ref{sec:SIDTHE}.

%% file: Sections/sidthe.tex
\section{The SIDTHE Compartmental Model}\label{sec:SIDTHE}
We now introduce the controlled SIDTHE model which will be used to describe the underlying pandemic. The model is given by a set of six nonlinear ODEs\footnote{All derivatives are taken with respect to time, but we drop the dependence on $t$ for simplicity.
}
\begin{subequations}
    \begin{align} 
      \dot{S} &= -\alpha (1 - u) S I        \label{eq:def_Sdot}   \\
      \dot{I} &=  \alpha (1 - u) S I - \gamma  \left (1 + \frac{\lambda}{\lambda + \gamma}  \right ) I \label{eq:def_Idot}\\
      \dot{D} &=  \gamma I - (\delta  +\lambda) D \\
      \dot{T} &=  \delta D - (\sigma + \tau) T\\
      \dot{H} &=  \sigma T + \lambda D + \lambda \left  (\frac{ \gamma }{\lambda + \gamma} \right ) I \\
      \dot{E} &=  \tau T,
    \end{align}
    \label{eq:SIDTHE}
\end{subequations}
alongside an initial condition $x_0 \in \Delta_6$.


According to~\eqref{eq:SIDTHE}, the population is divided into six compartments, each representing different stages of infection.
Compartment $S$ (susceptible) represents uninfected healthy individuals, $I$ (infected) includes asymptomatic, pauci-symptomatic and symptomatic undetected individuals; $D$ (diagnosed) gathers infected individuals who have been detected by means of positive test and recorded by authorities; $T$ represents both hospitalized population (not in life-threatening danger) and on life support intensive care individuals. The compartment $E$ represents deceased individuals as a consequence of the disease, while $H$ represents healed individuals. 

Concerning the parameters of the model, 
$\alpha$ represents the transmission rate due to contacts between susceptible and infected individuals, while $\gamma$ is the diagnosis parameter, which denotes the detection rate of individuals in $I$. The critical/aggravation parameter is denoted by $\delta$ and represents the rate at which detected individuals develop severe symptoms with in-hospital support needs. In addition, recovery rates for hospitalized individuals are denoted by $\sigma$. For the infected subpopulation that does not need hospital care ($I$ and $D$), a similar recovery rate based on $\lambda$ is taken into account. Finally, the mortality rate parameter $\tau$ dictates the flow from $T$ to $E$.
All the aforementioned parameters 
\begin{align}
\theta = [\alpha,\,\gamma,\,\lambda,\,\delta,\,\sigma,\,\tau]^\top,  
\label{eq:SIDTHE_params_vector}
\end{align}
are strictly positive.
The control action $u$ is supposed to belong to the set $\U =  [0, \umax]$, where $ 0 \le \umax  <1$ is a proxy for the severity of non-pharmaceutical interventions we can exert and is manifested as a reduction of $\alpha$, see \eqref{eq:def_Sdot}.
Normal societal functioning corresponds to $u = 0$, while the most restrictive lockdown corresponds to $u = \umax$. 
The inflow and outflow of people to and from various compartments is sketched in Figure~\ref{fig:sidthe}.
\begin{figure} 
    \centering
    \includegraphics[scale = 1.1]{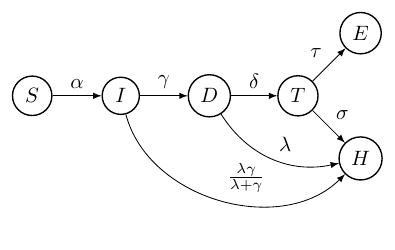}\vspace{-.5cm}
    \caption{A directed acyclic graph representation of the SIDTHE model.}
    \label{fig:sidthe}
\end{figure}
We also assume that people admitted to hospital care are detected in advance and that they are perfectly isolated.
Note that since, on average, we expect that the $D$ compartment is composed of individuals who are showing symptoms, the flow from that compartment to $H$ should be faster than the flow from $I$ to $H$, which we model with $\frac{\lambda \gamma}{\lambda + \gamma}$ outflow rate from $I$ compartment.

The main motivation behind the proposed model is to capture the dynamics of \emph{threatened} individuals, namely the portion of the population severely affected by the disease (and thus in need of hospitalization and/or intensive care) while retaining desirable theoretical properties like observability and, in particular, identifiably. This last property is particularly relevant as, if satisfied, it implies that one can uniquely determine the parameters of the model given publicly pandemic available datasets, without the need for resorting to (eventually not validated) assumptions on their values. It should be noted that this property is often not verified in complex compartmental models such as SIDARTHE~\cite{SIDARTHE}, thus invalidating the possibility of uniquely estimating their parameters from open-source data.
Meanwhile, the parameters of the model in \eqref{eq:SIDTHE} are structurally and practically identifiable, with the model meeting the practical identifiability test as outlined in~\cite{Kreutz2018} using publicly available data on the COVID-19 pandemic in Italy~\cite{GitHubCOVID19}.  

\subsection{Positive Invariance}
When trying to contain a pandemic outbreak, it is likely that one of the goals is to keep the hospitalized subpopulation $T$ below a 
threshold 
$T_{\max}$, ultimately preventing an ICU overload.
As the number of hospitalized people is the state $T$ of the dynamical system \eqref{eq:SIDTHE}, this goal can be expressed by defining a \textit{positive invariant} set within which the number of infected individuals remains below the prescribed threshold.
Indeed, any initial state of the pandemic that lies within this set is guaranteed to never leave it. 

To enforce this, we need to have $\dot T \leq 0$ whenever $T = T_{\max}$. However, this cannot be achieved by inspecting the department $T$ alone. In fact, we can see that the control action $u$ does not appear in the dynamics of $T$ in \eqref{eq:SIDTHE}, but it does affect it indirectly through the compartments $S,I$ and $D$. For this reason, the bound on $T$ implicitly imposes bounds on the upstream nodes $S$, $I$, and $D$.
The positive invariant set that we consider is then a box given by
\begin{align} 
    X_\mathrm{f} \!&:=\hspace{-1mm} [0, S_{\max}] \!\times\! [0, I_{\max}] \!\times\! [0, D_{\max}] \!\times\! [0, T_{\max}] \!\times\! [0,1]^2,
\label{eq:inv_set} 
\end{align}
where the first three upper bounds depend on $u_{\max}$ and $T_{\max}$ as follows:
\begin{subequations}
    \begin{align}
        S_{\max} &:= \frac{ \gamma}{\alpha (1 - u_{\max})} \left (1+\frac{\lambda}{\lambda + \gamma} \right ), \\
        I_{\max} &:= \frac{(\delta  +\lambda)(\sigma + \tau)}{\gamma\delta} T_{\max}, \\
        D_{\max} &:= \frac{(\sigma + \tau)}{\delta} T_{\max}.
    \end{align}
    \label{eq:terminal_set_individual}
\end{subequations}
To prove that the set $X_\mathrm{f}$ in \eqref{eq:inv_set} is a positive invariant set for the dynamics \eqref{eq:SIDTHE}, we introduce some key results and definitions, starting with a tangent cone to a set.

\begin{definition}[Tangent cone]
        The Bouligand tangent cone to a closed set $S$ at point $x$, denoted $T_S (x)$, is defined as
        \begin{align*}
            T_S(x) := \left\{v \in \Re^n : {\lim \inf}_{t \rightarrow 0^+} \frac{d_S(x + tv)}{t} = 0 \right\}.
        \end{align*}
\end{definition}
    
Intuitively, the tangent cone collects all the directions that, from $x$, point ``inside'' (or tangentially to) the set.
This definition allows us to present a fundamental result in set invariance, known as Nagumo's theorem.
Even though this theorem is stated for autonomous systems (without a control action), it can be easily extended to characterize controlled invariance as follows (see the discussion in \cite[\S3]{BLANCHINI_set_invariance}).
\begin{theorem}[Nagumo, 1942]
        Consider the system $\dot x = f(x,u)$, with $x \in \Re^n$, $u \in \Re^p$, and $f: \Re^{n+p} \rightarrow \Re^n$ Lipschitz continuous.
        Suppose that the input $u$ is selected as $u = \Phi(x)$ with $\Phi: \Re^n \rightarrow \Re^p$ also Lipschitz continuous.
        The closed, convex set $S$ is positively invariant for the system if and only if
        \begin{align}
            \label{eq:pos_inv_def}
            f(x, \Phi(x)) \in T_S(x), \quad \forall x \in \partial S,
        \end{align}
        where $\partial S$ denotes the boundary of the set $S$.
\end{theorem}
    
Using Nagumo's theorem, we can show that the set $X_\mathrm{f}$ in \eqref{eq:inv_set} is positively invariant for the controlled SIDTHE dynamics \eqref{eq:SIDTHE}.
This result is summarized in the following theorem.

\begin{theorem}
        Given the upper bound $u_{\max} > 0$ on the control action, it is always possible to select $u = \Phi(x) \in [0, u_{\max}]$ such that the set $X_\mathrm{f}$ in \eqref{eq:inv_set} is positively invariant for the SIDTHE dynamics \eqref{eq:SIDTHE}.
\end{theorem}

\begin{proof}
        Notice that $X_\mathrm{f}$ is a box set, therefore it is closed and convex.
        Given the particular shape of $X_\mathrm{f}$, denoting as $x_i, \; i \in \{1, \ldots, 6\}$ the $i$-th component of the state $x = (S,I,D,T,H,E)$, \eqref{eq:pos_inv_def} is equivalent to
        \begin{align*}
            \begin{cases}
                x_i = 0 \implies \dot x_i \geq 0, \\
                x_i = x_i^{\max} \implies \dot x_i \leq 0
            \end{cases} \quad \quad i \in \{1, \ldots, 6\}.
        \end{align*}
        Some of these conditions 
        are implicitly enforced by the compartmental nature of our model, see \eqref{eq:prob_simplex}.
        Indeed, if we consider any of the equations $\dot x_i$ in \eqref{eq:SIDTHE} and we substitute $x_i = 0$, we immediately obtain $\dot x_i \geq 0$ from the non-negativity of the other state components.
        Similarly, we can trivially check the upper bound condition for $H$ and $E$. Since the compartmental model enforces $1^\top x = 1$ (see \eqref{eq:prob_simplex}), any state component being equal to $1$ implies that the others are $0$.
        Consequently, $\dot H$ and $\dot E$ are $0$ whenever the corresponding compartment is $1$.
        
        It remains to check that, when $S,I,D,T$ are equal to their upper bounds, their derivative is negative.
        Clearly, $\dot S \leq 0$ for any non-negative $S,I$.
        The dynamics of $\dot I$, given the upper bound on $S$, yields
        \begin{align*}
            \dot I \leq I_{\max} \left( \frac{\gamma (1-u)}{1 - u_{\max}} \left( 1 + \frac{\lambda}{\lambda+\gamma} \right) - \gamma \left( 1 + \frac{\lambda}{\lambda+\gamma} \right) \right).
        \end{align*}
        Hance, any selection of $u=\Phi(I,\cdot)$ such that $\Phi(I_{\max}, \cdot) = u_{\max}$ results in $I = I_{\max} \implies \dot I \leq 0$, as needed.
        By substituting $I \leq I_{\max}$ and $D = D_{\max}$ in the dynamics of $\dot D$, we obtain
        \begin{align*}
            \dot D \leq T_{\max} \left( \gamma \frac{(\delta+\lambda)(\sigma+\tau)}{\gamma\delta} - (\delta+\lambda) \frac{\sigma+\tau}{\delta} \right) = 0.
        \end{align*}
        Similarly, substituting $D \leq D_{\max}$ and $T = T_{\max}$ yields $\dot T \leq 0$, completing the proof.
\end{proof}
 
\begin{remark}
     Although the design of the positive invariant set $X_\mathrm{f}$ in \eqref{eq:inv_set}-\eqref{eq:terminal_set_individual} does not take into account the uncertainty of the parameters, it is possible to obtain a set $X_\mathrm{f}$ in \eqref{eq:inv_set} that is common to all different combinations of parameter values.
     To do so, supposing that each parameter $\theta$ belongs to an interval $[ \theta_{\min}, \theta_{\max}]$, the upper bounds in \eqref{eq:terminal_set_individual} need to be computed by considering $\theta = \theta_{\min}$ whenever it appears in the numerator and considering $\theta = \theta_{\max}$ if it appears in the denominator.
\end{remark}

We can further establish that the values of $I$,$D$, and $T$ decay exponentially to zero in $X_\mathrm{f}$ following the approach in~\cite{MPC_teminalcost}.
Due to the monotonicity of $S$, we can overapproximate the solution of our system 
by considering a surrogate one  where $S(0)$ is kept constant.
In that case, the state evolution of the over-approximating sub-system is given by
\begin{align}
    \begin{bmatrix}
        \dot{I}\\
        \dot{D} \\
        \dot{T}
    \end{bmatrix}
=
\underbrace{\begin{bmatrix}
    \bar{\alpha} S(0) - \frac{\gamma(2\lambda + \gamma)}{\lambda + \gamma} & 0 & 0 \\
    \gamma & -(\lambda + \delta) & 0 \\
    0 & \delta & -(\sigma + \tau) 
\end{bmatrix}}_{F \left(S(0),\lambda, \gamma, \delta, \sigma, \tau\right)}
\begin{bmatrix}
    I\\
    D \\
    T
\end{bmatrix}
\label{eq:overapprox}
\end{align}
where $\bar{\alpha} = \alpha(1 - u)$.
The explicit solution to~\eqref{eq:overapprox} is then given by the matrix exponential 
\begin{align}
[I, \, D, \, T]^\top = e^{F \left(S(0),\lambda, \gamma, \delta, \sigma, \tau\right) t} [I_0, \, D_0, \, T_0]^\top.
\end{align}
Since $F \left(S(0),\lambda, \gamma, \delta, \sigma, \tau\right)$ is lower triangular, its eigenvalues are the elements of the main diagonal, and their negativity is equivalent to exponential decay to zero.
Clearly, the second and third eigenvalue are negative, as all parameters are strictly positive, while imposing the negativity of the first element of the main diagonal yields
\begin{align}
    S(0) < \frac{ \gamma( 2\lambda  + \gamma) }{\alpha (1 - u)(\lambda + \gamma)},
    \label{eq:R0condition}
\end{align}
where the right-hand side is the $R_0^{-1}$ of our model.
Indeed, from the dynamics of $I$ given by \eqref{eq:def_Idot}, we can see that $I$ increases whenever $S R_0 > 1$.
Note that our expression for $R_0$ coincides with the next generation matrix calculation~\cite{nextgenmatrix}. 
Meanwhile, the upper bound $S_{\max}$ in \eqref{eq:terminal_set_individual} corresponds to the maximum possible value for $R_0^{-1}$, given the constraints on the control action $u$.

%% file: Sections/nmpc.tex
\section{Scenario based stochastic Nonlinear Model Predictive Control}
To set the ground for the description of the Scenario-based stochastic Nonlinear Model Predictive Control (SBSNMPC) algorithm we propose for pandemic control, we first describe the principles of Nonlinear Model Predictive Control (NMPC)~\cite{NMPC_grune}. At its core, the standard NMPC uses a nonlinear mathematical model of the process (virus transmission in our case) to predict its future evolution and choose the best possible control action, given a certain objective. It does so by solving a finite horizon optimization problem:
\begin{subequations}
    \begin{align}  
        \minimize_{x,u} \,      & \, \int_0^{t_f}  \ell(x(t), u(t) ) dt \\
        \stt  \quad \dot{x}(t) & = f(x(t), u(t), \theta)       \\
        g(x(t),u(t) )               & \le 0 ,      \\
        g_f(x(t) )             &  \le 0,           \\
        x(0)                & = x_0,                       
    \end{align}
    \label{eq:Lagrange_problem}
\end{subequations} 
at each time step. Once the first control input is applied and the new state is measured (or estimated), the algorithm is restarted. 
The canonical finite-horizon problem without terminal cost in~\eqref{eq:Lagrange_problem} is a generalization of the problem we aim to solve to tackle pandemic control. In particular, we assume that the cost function $\ell : \mathcal U \rightarrow \Re_+$ is convex on its domain and solely depends on the input.
This assumption rests on the observation that more stringent NPIs entail higher societal cost and that the relationship between them is often very steep. A linear cost function, used by many authors, satisfies this assumption, but we argue that more general expressions for the cost could provide more realistic representations of the real world. 

Note that \eqref{eq:Lagrange_problem} 
can be solved by indirect or direct methods. We will opt for the former option, transcribing the infinite-dimensional problem to its discrete, finite-dimensional nonlinear program version, which is then solved using off-the-shelf solvers such as~\cite{IPOPT}. To this end, the state update function $f(\cdot)$ is discretized into $f_d(\cdot)$, and the continuous-time integral is approximated by a finite sum. In this paper, this discretization is performed by adopting the multiple shooting approach, choosing fourth-order Runge–Kutta (RK4) as the integration method.

\subsection{Scenario-based NMPC}
The standard NMPC procedure previously outlined relies on a deterministic outlook on reality, not considering possible uncertainties causing differences between the predicted and the actual evolution of the controlled system. 

Unlike standard MPC, scenario-based model predictive control
considers multiple realizations of the parameter
of the prediction model and system states, 
making the control inputs more robust to parameter uncertainty. 
We deem this feature important when dealing with 
pandemic management, given the well-known challenges of exactly modeling pandemic evolution and the variability between the predicted and actual evolution of the pandemic also dictated by (often highly unpredictable) human behaviors.
Hereafter, we assume that the parameter vector $\theta$ (see \eqref{eq:SIDTHE_params_vector}) can take $N_S$ different values, or \emph{scenarios}, from a set $\Theta$. We will denote the $i$-th scenario by $\theta^i$. We follow this convention for other variables as well, hence state trajectory $x^i$ corresponds to the state evolution according to the control trajectory $u^i$ and the parameter scenario $\theta^i$. Note that, in what follows, we assume that a probability vector $p \in \Delta_{N_S}$, describing the probability of each scenario, is known.

\begin{figure}[!tb]
    \centering
    \includegraphics[scale = 0.65]{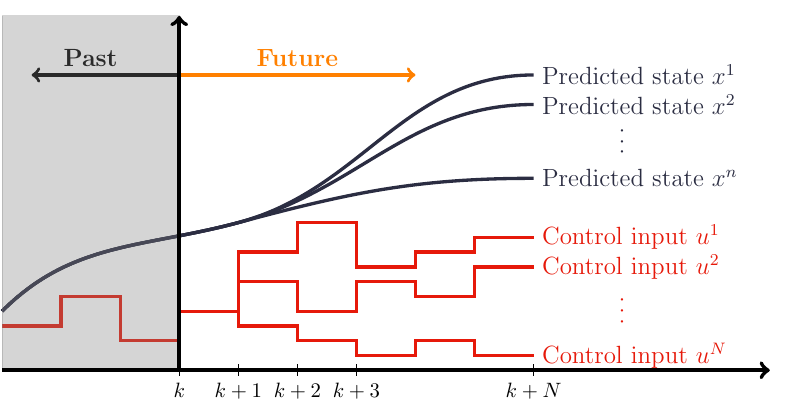}
    \caption{Depiction of the basic idea of scenario-based MPC. Note that 
    all the control trajectories have to agree on the first control action due to causality constraint}
    \label{fig:scenarioMPC}
\end{figure}
Based on these assumptions, the working principle of scenario-based SNMPC 
is depicted in Figure~\ref{fig:scenarioMPC}.
At time $k$, we estimate the current state $x_0$ and predict $N_S$ different possible trajectories depending on parameters $\theta^i$ and control input sequences $u^i$. Note that $u^i = \{u^i_k\}_{k \in [0,N-1]}$ is a sequence (or trajectory) of decision variables, where the index $k$ is the discrete time step, and $N$ is the prediction horizon. These decision variables are optimized by solving the following  problem 
\begin{subequations}\label{eq:mpc_discrete}
    \begin{align}
        \minimize_{ x^i, u^i } \,  \quad     \, \sum_{i = 0}^{N_S}& \sum_{k = 0}^{N - 1} p_i \, \ell( u_k ) \\
        \stt  \quad x_{k+1}^i & = f_d(x^i_k, u^i_k, \theta^i)                               \\
        u_k^i                 & \in [0,u_{\text{max}}]                                    \\
        Ax_k^i                & \le c_{\text{max}}                                        \\
        u_0^0 - u_0^i &= 0 \\
        x_0^i - \bar{x_0} &= 0 \\
        \forall i &\in \{1,\dots,N_S\}, k \in \{1, N\} \nonumber 
    \end{align}
    \label{eq:pandemic_problem}
\end{subequations}
where $A$ is a selection matrix used to select the state trajectory that describes the compartment of interest and $c_{\text{max}}$ is a constant. Note that this problem is non-convex due to the typical structure of compartmental models, which involves bilinear terms in the dynamics (see~\eqref{eq:def_Sdot},\eqref{eq:def_Idot}).

Toward assessing the impact of different strategies to tackle uncertainties in pandemic control,
we consider two flavors of SBSNMPC which differ in how we select the control trajectories. In the first (referred to as the \textit{robust} approach), we have a single shared control trajectory for all realizations of the uncertainty vector $\theta$. The problem in \eqref{eq:pandemic_problem} is thus augmented with the following constraints
\begin{align}
u_k^0 = u_k^i, \quad \forall i \in \{1,\dots,N_S\}, \forall k \{1,\dots, N\}.
\end{align}
In the second (the \textit{recourse} approach), we assign a dedicated control trajectory to each realization of uncertainty.
The robust approach, in the spirit of~\cite{allgower_MPC_covid}, is more conservative than the recourse one. This can be easily seen, as the robust approach can be obtained from the recourse one, by requiring all the control trajectories to be the same, i.e., by introducing additional equality constraints. Nonetheless, we argue that the recourse approach, in addition to being less conservative, is also more informative for policymakers, as it provides insight into the required actions for each realization of the uncertain vector $\theta$.
Even though we end up with a larger optimization problem, we do not
consider it a hindrance. In fact, pandemic policies are usually updated every few weeks, leaving plenty of time to solve the optimization problem offline. 

%% file: Sections/numerical.tex
\section{Numerical simulations}
In this section, we highlight some basic properties of our approach 
via a simulation study. To this end, we discretize the continuous system dynamics with RK4, 
with sampling time of one day. Note that, as the stage cost function depends only on the control actions (NPIs) which are piece-wise constant, 
the evaluation of the integral is exact in its discrete-time version. The cost function is selected as $\ell(u) = u^2$ and 
maximal hospital capacity is set to $T_{\max} = 0.002$, i.e. $0.2 \%$ of the total population. The maximum reduction in virus transmissibility is set to $\umax = 0.75$ and NPI measures are kept constant for $T_{\textrm{NPI}} = 14$ days. 
All simulations were implemented using Julia~\cite{Julia-2017}.

\subsection{Scenario-based SNMPC}
We consider a SIDTHE model given by~\eqref{eq:SIDTHE} starting from the initial condition
$$x_0 = [0.99,\;0.008,\,1.9\times 10^{-4},\,1\times 10^{-4},\,0,\,0]^\top, $$ 
choosing the following nominal values for the SIDTHE's parameters 
\begin{align}
\theta = [0.35,\;0.1,\,0.09,\,2\times 10^{-3},\,0.015,\,0.01]^\top,
\label{eq:example_params}
\end{align}
with the parameters' ordering given in~\eqref{eq:SIDTHE_params_vector}.
To obtain different scenarios, we consider perturbations of $\pm 5\%$ on each of these parameters 
resulting in three possible values, i.e., $0.95 \cdot \theta_i,$ $\theta_i$ and $1.05 \cdot \theta_i$.
This yields $N_S = 3^6=729$ possible scenarios for the SIDTHE parameters realization.

\begin{figure}[ht]
    \centering
     \includegraphics[scale =0.9]{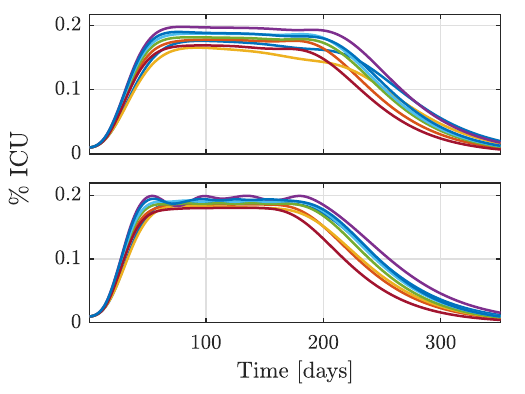}\vspace{-.3cm}
    \caption{Dynamics of $T$-\textit{Threatened} compartment with two approaches 
    in the closed-loop, \textit{robust} (top) and \textit{recourse} (bottom).}
    \label{fig:states}
\end{figure}

\begin{figure}[ht]
    \centering
    \includegraphics[scale = 0.9]{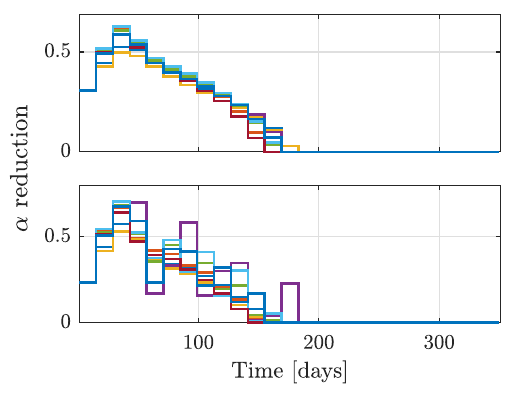}\vspace{-.3cm}
    \caption{Control action on the transmission rate $\alpha$ computed with \textit{robust} (top) approach and \textit{recourse} (bottom) approach.}
    \label{fig:control}
\end{figure}

Figure~\ref{fig:states} compares the evolution of the compartment $T$, for different parameter realizations, using the robust approach (top plot) and the recourse approach (bottom plot).
We can observe that the scenario-based MPC always leaves a safety margin to account for the possibility of a worst-case scenario occurring within the prediction horizon. Moreover, recourse MPC is shown to be slightly less conservative as it leaves less of a margin of safety for the same disturbances compared to the robust one. The corresponding trajectories of the NPI severity (i.e., the control input) are depicted in Figure~\ref{fig:control}.
Here, again, we note that recourse MPC leads to a slightly reduced control input on average, i.e., less strain on the population. 
Another difference lies in the first control actions for both approaches, where the recourse strategy has a significantly lower control input.

To illustrate the working principle of scenario-based MPC, we simulated the (controlled) pandemic outbreak until day $98$. Figures~\ref{fig:state_predictions}-\ref{fig:control_pred} depict a snapshot of the MPC algorithm predictions at this time instant.
Figure~\ref{fig:state_predictions}, in particular, shows how the value of the compartment $T$ approaches the threshold, without ever going over it for any of the possible scenarios.
This is in line with our goal, which was to contain the pandemic outbreak while minimizing the societal burden (corresponding to the intensity of the control actions, shown in Figure~\ref{fig:control_pred}).


Lastly, Figure \ref{fig:InfeasTraj} shows the behavior of a standard (nominal) MPC algorithm implemented in our case study. We can see that this algorithm cannot handle \emph{bad} scenarios.
The predictions of the standard MPC are calibrated using the nominal parameter values~\eqref{eq:example_params}, without taking any uncertainty into account. Therefore, with some scenarios, it may happen that no control action can keep $T$ from going over $T_{\max}$, resulting in overrun of hospital capacity.

\begin{figure}[t]
    \centering
    \includegraphics[scale = 0.9]{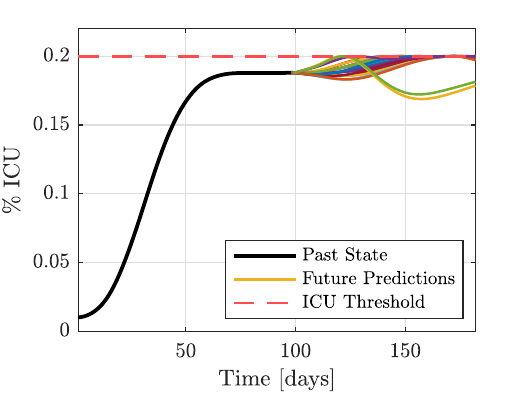}\vspace{-.3cm}
    \caption{Subset ($25$ of $3^6$) of different predicted evolutions for the \textit{T} compartment  at day $98$. The black curve depicts the  trajectory of $T$ on days $1$-$97$.}
    \label{fig:state_predictions}
\end{figure}

\begin{figure}[ht]
    \centering 
    \includegraphics[scale = 0.9]{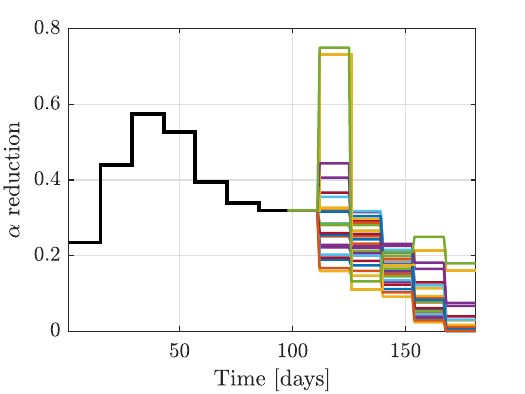}\vspace{-.3cm}
    \caption{Subset ($25$ of $3^6$) of different predicted control actions at day $98$. The black curve depicts the actual NPIs applied on days $1$-$97$.}
    \label{fig:control_pred}
\end{figure}

\begin{figure}[ht]
    \centering 
    \includegraphics[scale = 0.94]{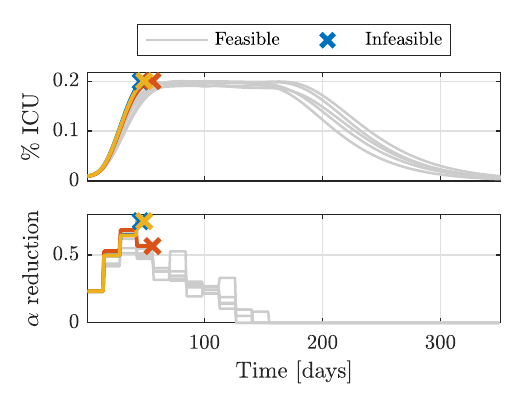}\vspace{-.3cm}
    \caption{Example of a vanilla MPC  infeasibility under parameter uncertainty. Some parameter sets yield a feasible solution, but this simple approach fails to find one for \emph{``bad''} scenarios}
    \label{fig:InfeasTraj}
\end{figure}
\subsection{Safe Set Operation}
We now focus on highlighting the effects of operating within the positive invariant set for all possible disturbances.
We chose a safe set $X_\mathrm{f} = \cap_{i}^{N_S} X_\mathrm{f}(\theta^i)$, as the intersection of all positively invariant sets for every scenario $\theta_i$ calculated
as in~\eqref{eq:inv_set}-\eqref{eq:terminal_set_individual}.
Hence, $X_\mathrm{f}$ is positively invariant for any possible scenario realization and is given by~\eqref{eq:inv_set} with $S_{\max} = 1$,
$I_{\max} = 0.0188$, $D_{\max} = 0.0226$, $T_{\max} = 0.002$. 
Note that, starting from the same initial condition as considered in our previous simulation it holds $x_0 \in X_\mathrm{f}$.

Next, we simulate a recourse MPC controller with nominal parameters and constraint states $S,I,D$ and $T$ to be inside $X_\mathrm{f}$ at all times.
The role of MPC here is to minimize the societal cost while operating
within the safe set, as naive application of $\umax$ would be too conservative. Note that MPC is recursively feasible since $\umax$ is a viable control action to keep us within the same set at all times for any point in $X_\mathrm{f}$.

\begin{figure}[ht]
    \centering 
    \includegraphics[scale = 0.92]{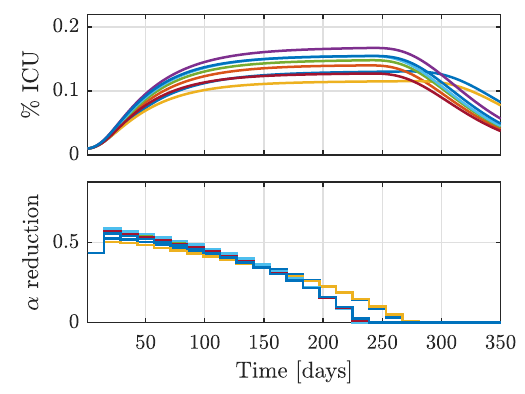}\vspace{-.3cm}   
    \caption{Recourse MPC constrained to operate within the safe set $X_\mathrm{f}$.}
    \label{fig:MPC_SAFE}
\end{figure}
 
Comparing figure~\ref{fig:MPC_SAFE} to the bottom figures in \ref{fig:states}-\ref{fig:control} clearly shows conservativeness introduced by begin safe at all times. In the future, we plan to make use of this terminal set property to make our MPC scheme recursively feasible at all times,
while still allowing some states to leave this set with the goal of improving performance.

%% file: Sections/conclusions.tex
\section{Conclusion}
This paper presented a control framework for pandemic response that optimizes the level of non-pharmaceutical interventions (NPIs) by accounting for multiple possible pandemic evolutions.
To this end, we devised a compartmental model that strikes a 
balance between identifiability and level of detail, allowing us to explicitly account for hospital pressure
. We embedded this model into a scenario-based Model Predictive Control (MPC) scheme that optimizes decisions by considering different possible scenarios. This crucial feature allows our control algorithm to outperform standard MPC methods in terms of safety.
For our future work, we envision the development of a data-driven online estimation scheme to track the changes in model parameters, as well as the derivation of conditions for the recursive feasibility of the presented scenario-based MPC by using the developed invariant set.
